\documentclass[12pt,onecolumn,twoside]{IEEEtran}
\usepackage{amsmath}
\usepackage{amsthm}
\usepackage{amsfonts}
\usepackage{graphicx}
\usepackage{latexsym}
\usepackage{amssymb}
\usepackage{color}
\usepackage{url}

\usepackage{xfrac}

\DeclareMathAlphabet{\mathbfsl}{OT1}{ppl}{b}{it} 

\makeatletter
\DeclareRobustCommand{\nsbinom}{\genfrac[]\z@{}}
\makeatother

\newcommand{\abs}[1]{\left|#1\right|}

\newcommand{\mmod}{{\mbox{mod}}}

\newcommand{\field}[1]{\mathbb{#1}}

\newcommand{\Z}{\field{Z}}

\newcommand{\F}{\field{F}}

\newcommand{\cA}{{\cal A}}

\newcommand{\cD}{{\cal D}}

\newcommand{\cL}{{\cal L}}

\newcommand{\cT}{{\cal T}}

\newcommand{\bE}{{\bf E}}
\newcommand{\bF}{{\bf F}}

\newcommand{\bfzero}{{\bf 0}}
\newcommand{\bfone}{{\bf 1}}
\newcommand{\bftwo}{{\bf 2}}
\newcommand{\bfthree}{{\bf 3}}
\newcommand{\bffour}{{\bf 4}}

\newcommand{\linadd}{\kern1pt\mbox{\small$\boxplus$}\kern1pt}

\newtheorem{definition}{Definition}
\newtheorem{theorem}{Theorem}
\newtheorem{lemma}{Lemma}

\newtheorem{example}{Example}

\theoremstyle{definition}

\begin{document}

\bibliographystyle{plain}

\title{Optimal Non-Binary Single-Track Gray Code}

\author{
{\sc Tuvi Etzion}\thanks{Department of Computer Science, Technion,
Haifa 3200003, Israel, e-mail: {\tt etzion@cs.technion.ac.il}.
}
}

\maketitle

\begin{abstract}
A single-track Gray code is a cyclic Gray code with codewords of length $n$, over an alphabet of size $m$, such that all the $n$ tracks
that correspond to the $n$ distinct coordinates of the codewords are cyclic shifts of the first track.
Such codes have advantages over the conventional Gray codes in certain quantization and coding applications.
Unless $n=2$, there are no such binary codes that contain all the $2^n$ codewords of length $n$.
In this paper, we prove that such codes of length $p^t$, $n \geq 2$, with $p^{p^t}$ codewords, over $\F_p$, $p$ prime, exist,
for $p=3$ and $p=5$. For larger prime $p$ an appropriate code for $t=2$, implies the existence of such a code for any $t >2$.
If the alphabet size $m$ is not a prime there are also indications that such codes exist.
\end{abstract}

\vspace{0.5cm}


\vspace{0.5cm}



\newpage
\section{Introduction}

A {\bf \emph{length $n$ Gray code}} is a list of distinct words of length $n$ (called the codewords) over an alphabet $\Sigma$,
$$
W_0, W_1,\ldots, W_{P-1},
$$
having the property that any two adjacent codewords (including the last and the first)
$W_i$, $W_{i+1}$ differ in exactly one component. This Gray code has {\bf \emph{period}} $P$.
Gray codes were found by Gray~\cite{Gra53} and introduced by Gilbert~\cite{Gil58}.
Generalizations of Gray codes were given over the years. Such generalizations include listing subsets of the binary $n$-tuples
in a Gray code manner, in such a way that the list has some more pre-specified properties. These properties were
usually forced by a specific application for the Gray code. As an example, we have the uniformly balanced Gray codes.
In certain applications, it is required that the number of bit changes will be uniformly distributed among the bit positions.
Uniformly balanced Gray codes were shown to exist for $n$ which is a power of $2$ by Wagner and West~\cite{WaWe94}.
Bhat and Savage~\cite{BhSa96} have shown that such codes exist for all $n$. During the years Gray codes and their generalizations have
found applications in a variety of areas such as information storage and retrieval~\cite{CCC92},
processor allocation in the hypercube~\cite{ChSh90}, various imaging systems~\cite{CZFYZ,ZPAC},
statistics~\cite{DiHo94}, codes for certain memory devices~\cite{Etz92},
hashing~\cite{Fal88}, puzzles, such as the Chinese Rings and Tower
of Hanoi~\cite{Gar72}, rank-modulation for flash memories~\cite{HoEt14,JMSB}, ordering of documents on shelves~\cite{Los92}, signal
encoding~\cite{Lud81}, data compression~\cite{Ric86}, circuit testing~\cite{RoCo81}, and distributed video coding~\cite{SSWZS}.
An early survey on Gray codes can be found in~\cite{Sav97} and an ongoing excellent survey in~\cite{Mut23}.

Finally, we define single-track Gray codes, which are the topic for the current exposition.
A length $n$, Period~$P$, {\bf \emph{single-track Gray code}} (in short, STGC) is a list of~$P$ distinct binary words of length~$n$, such that
every two consecutive words, including the last and the first, differ in exactly one position. When
looking at the list as a $P \times n$ array (or an $n \times P$ array),
each column (row, respectively) of the array is a cyclic shift of the first column (row, respectively).
These codes were defined first by Hiltgen, Paterson, and Brandestini~\cite{HPB96}
and considered later in~\cite{EtPa96,FaLi23,QYYS,ScEt99,YLWLLZ,ZZBLZ} for their theoretical value and for their various applications.
They were also considered in a new book~\cite{Etz24}.

The first application for these codes was given in~\cite{HPB96}.
A length $n$, period $P$ Gray code can be used to record the absolute angular positions of a rotating wheel by encoding
(e.g., optically) the codewords on concentrically arranged tracks. Then reading heads, mounted in parallel across
the tracks, suffice to recover the codewords. When the heads are nearly aligned with the division between two codewords,
any components which change between those words will be in doubt and a spurious position value may result. Such
quantization errors are minimized by using a Gray code encoding, for then exactly one component can be in doubt
and the two codewords that could possibly result identify the positions bordering the division, resulting in a small angular
error. When high resolution is required, the need for a large number of concentric tracks results in encoders with large
physical dimensions. This poses a problem in the design of small-scale or high-speed devices. Single-track Gray codes
were proposed in~\cite{HPB96} as a way of overcoming these problems. Note, that since all the columns in these codes are cyclic shifts
of the first one, it follows that the code is also a uniformly balanced Gray code, which again can be described by a single column.

The main goal is to construct a length $n$, Period $P$ STGC, where $P$ is large as possible.
Bounds on $P$ can be found in~\cite[Lemma 2]{HPB96} for binary codes.
Of special interest are codes that contain all the words of length $n$. Such a code will be called a {\bf \emph{full-period STGC}}.
If $P=2^n$, then only when $n$ is a power of 2 such a binary code can exist.
However, it was proved in~\cite{ScEt99} that such an STGC exists only for $n=2$, When $n=2^t$ is larger than 2,
the largest possible period of an STGC is $2^t -2t$. Such codes were constructed in~\cite{EtPa96}.
To obtain codes with large periods, the two constructions presented
in the following two theorems are used. For these theorems, we need the following definitions.

For a cyclic sequence $S = [s_0 s_1 ~ \cdots ~ s_{k-1}]$,
the {\bf \emph{the shift operator}} is defined by
$$
\bE [s_0 s_1 ~ \cdots ~ s_{k-1}] = [s_1 ~ \cdots ~ s_{k-1} s_0] ~.
$$
Two cyclic sequences $S_1$ and $S_2$ are the same sequence and called {\bf \emph{equivalent}}, $S_1 \simeq S_2$, if $S_2$ is
a cyclic shift of $S_1$.

\begin{theorem}
\label{thm:single_necklaces}
Let $S_0, S_1,\ldots,S_{r-1}$ be $r$ binary pairwise inequivalent sequences of period $n$,
such that for each $i$, $0 \leq i < r-1$, $S_i$ and $S_{i+1}$ differ in exactly
one coordinate. If there also exists an integer $\ell$, where $\gcd(\ell,n)=1$,
such that  $S_{r-1}$ and $\bE^{\ell} S_0$ differ in exactly one coordinate,
then the following words (read row by row) form a length $n$, period~$nr$ STGC,
$$
\begin{array}{lclcccl}
S_0 & \hspace{0.3cm} & S_1 & \hspace{0.3cm} & \cdots & \hspace{0.3cm} & S_{r-1} \\
\bE^\ell S_0 & \hspace{0.3cm} & \bE^\ell S_1 & \hspace{0.3cm} & \cdots & \hspace{0.3cm} & \bE^\ell S_{r-1} \\
\bE^{2\ell} S_0 & \hspace{0.3cm} & \bE^{2\ell} S_1 & \hspace{0.3cm} & \cdots & \hspace{0.3cm} & \bE^{2\ell} S_{r-1} \\
\vdots & \hspace{0.3cm} & \vdots  & \hspace{0.3cm} & \vdots & \hspace{0.3cm} & \vdots \\
\bE^{(n-1)\ell} S_0 & \hspace{0.3cm} & \bE^{(n-1)\ell} S_1 & \hspace{0.3cm} & \cdots & \hspace{0.3cm} & \bE^{(n-1)\ell} S_{r-1} \\
\end{array}~.
$$
\end{theorem}

The second construction uses the notion of self-dual sequences.
A binary cyclic sequence $S = [s_0 s_1 ~ \cdots ~ s_{k-1}]$ is called a {\bf \emph{self-dual sequence}} (in short, SDS) if it
is equal to its complement, i.e., it is invariant under complement. Such sequences were handled for example in~\cite{Etz87}.

\begin{theorem}
\label{thm:self_dualST}
Let $S_0, S_1,\ldots,S_{r-1}$ be $r$ binary self-dual pairwise inequivalent
sequences of period~$2n$. For each $i$, $0 \leq i \leq r-1$, let $S_i =[ s_i^0 , s_i^1,\ldots, s_i^{2n-1}]$ and define
$$
\bF^j S_i = [ s_i^j, s_i^{j+1} ,\ldots, s_i^{j+n-1} ],
$$
where superscripts are taken modulo $2n$.

If for each $0 \leq i < r-1$, $S_i$ and $S_{i+1}$ differ in exactly
two coordinates, and there also exists an integer $\ell$, where $\gcd(\ell,2n)=1$,
such that  $S_{r-1}$ and $\bE^{\ell} S_0$ differ in exactly two coordinates,
then the following words form a length $n$, period $2nr$ STGC,
\vspace{-0.1cm}
$$
\begin{array}{lclcccl}
\bF^0 S_0 & \hspace{0.15cm} & \bF^0 S_1 & \hspace{0.15cm} & \cdots & \hspace{0.15cm} & \bF^0 S_{r-1} \\
\bF^\ell S_0 & \hspace{0.15cm} & \bF^\ell S_1 & \hspace{0.15cm} & \cdots & \hspace{0.15cm} & \bF^\ell S_{r-1} \\
\bF^{2\ell} S_0 & \hspace{0.15cm} & \bF^{2\ell} S_1 & \hspace{0.15cm} & \cdots & \hspace{0.15cm} & \bF^{2\ell} S_{r-1} \\
\vdots & \hspace{0.15cm} & \vdots  & \hspace{0.15cm} & \vdots & \hspace{0.15cm} & \vdots \\
\bF^{(2n-1)\ell} S_0 & \hspace{0.15cm} & \bF^{(2n-1)\ell} S_1 & \hspace{0.15cm} & \cdots & \hspace{0.15cm} & \bF^{(2n-1)\ell} S_{r-1} \\
\end{array} .
$$
\end{theorem}

The goal of the current work is to construct non-binary STGCs of length $m^t$ and period $m^{m^t}$ over an alphabet of size $m$.
It was proved in~\cite{ScEt99} that there is no such STGCs over $\F_p$, $p$ prime, if the difference between two consecutive
codewords on the different coordinate is always the same constant of $\F_p$. However, when this difference between consecutive
codewords might differ, such an STGC might exist, as will be proved in the current work.

The rest of the paper is organized as follows.
Section~\ref{sec:SDS} presents the non-binary self-dual sequences which will be the building for the non-binary STGCs.
In particular, a recursive construction for these sequences will be presented.
The second construction of binary STGCs is adapted for non-binary STGCs in Section~\ref{sec:STGC_basics}, and implementation of
the construction with SDSs for an alphabet of size $4$ is demonstrated.
The construction of STGCs is implemented later on alphabets whose size
is a prime. A recursive construction for these codes of length $p^t$ and period $p^{p^t}$,
where $p$ is an odd prime and $t \geq 3$ is provided. The construction works if there
exists an appropriate arrangement of all the SDSs of period $p^2$ (or period $27$ for $p=3$).
Such an arrangement is provided for $p=3$ and $p=5$.
Conclusions and directions for future research are provided in Section~\ref{sec:conclude}.

\section{Non-binary Self-Dual Sequences}
\label{sec:SDS}

The building blocks for our recursive construction of STGCs are SDSs. Three different constructions to generate such
sequences will be discussed in the paper: two recursive ones and one direct construction.
But, two of these construction will be only used to find seeds for the recursive construction.
Therefore, in this section only one recursive construction will be discussed. It will be discussed
only for alphabet whose size is an odd prime. It will be mainly demonstrated for alphabets of size three and
five before referring to an arbitrary alphabet of an odd prime size.

We start with some basic definitions on sequences.
A cyclic sequence $S = [s_0,s_1,\ldots,s_{k-1}]$ has {\bf \emph{length}} $k$.
A word (acyclic sequence) $X=(x_0,x_1,\cdots,x_{k-1})$ has also length $k$. The {\bf \emph{period}} of a cyclic sequence $S$
is the length of the shortest word $X$ such that $S=[X,X,\ldots,X]$. The period of acyclic sequences is not considered.

A very important operator with whom many of the properties of sequences can
be represented is the shift operator $\bE$. This operator shifts the whole sequence to the left.
For a word $(s_0 ,s_1 ,~ \cdots ~, s_{k-1})$ of length $k$,
the shift operator $\bE s_i = s_{i+1}$, where $0 \leq i < k-1$ and
for a cyclic sequence $S = [s_0 s_1 ~ \cdots ~ s_{k-1}]$, $\bE s_i = s_{i+1}$, where $0 \leq i < k-1$, and $\bE s_{k-1} = s_0$.
When the operator is applied on the whole cyclic sequence $S$ we have that
$$
\bE [s_0 s_1 ~ \cdots ~ s_{k-1}] = [s_1 ~ \cdots ~ s_{k-1} s_0] ~.
$$
When the operator is applied to an acyclic sequence $S$ we have that
$$
\bE (s_0 s_1 ~ \cdots ~ s_{k-1}) = (s_1 ~ \cdots ~ s_{k-1}) ~.
$$
We note that when $p$ is a prime we have that
$$
(\bE - \bfone)^{p^k} = \bE^{p^k} -1, ~~ k \geq 0 .
$$

In the sequel we denote by $\bfzero$ a vector of \emph{zeroes}, by $\bfone$ a vector of \emph{ones}, by $\bftwo$ a vector of \emph{twos}, etc.
However, $\bfone$ will also denote the identity operator which leaves a sequence unchanged. By this notation we have $\bE^0 =\bfone$.
It should be understood from the context which $\bfone$ is used.

The definition for a binary SDS is very simple. It is slightly more complicated for non-binary sequences,
where we use the following definition.

\begin{definition}
\label{def:SDSZm}
A non-binary cyclic sequence $S$ over $\Z_m$ is an SDS if $\bfone +S$ is equivalent to $S$.
\end{definition}

Definition~\ref{def:SDSZm} could be expanded for sequences invariant under addition of another constant different from \emph{one}.
Such a definition is discussed in~\cite{Etz26,Etz27}, but for construction of non-binary STGCs Definition~\ref{def:SDSZm}
is sufficient.

\begin{lemma}
\label{lem:periodSDS}
An SDS of length $p^n$ over $\F_p$, $p$ and prime, has period $p^n$.
\end{lemma}
\begin{proof}
Let $S=[X,X +\bfone ,X +\bftwo,\ldots,X +(p-1) \cdot \bfone]$ be an SDS of length $p^n$, where $X$ is a word of length $p^{n-1}$.
Since powers of $p$ are the only divisors of $p^n$, it follows that the period of $S$ is a divisor of $p^n$.
If $S$ has a period less than $p^n$, then $X+\bfone = X$, a contradiction.
\end{proof}

\begin{lemma}
\label{lem:enumerateSDS}
The number of inequivalent SDSs of period $p^n$, over $\Z_p$, $p$ prime, is $p^{p^{n-1} -n}$.
\end{lemma}
\begin{proof}
Each word $X$ of length $p^{n-1}$ defines an SDS $S=[X,X +\bfone ,X +\bftwo,\ldots,X +(p-1) \cdot \bfone]$ and each
word of length $p^{n-1}$ on $S$ defines the same sequence $S$. This implies that two such SDSs of period $p^n$
yield two disjoint sets of $p^n$ words of length $p^{n-1}$. Hence, the total number of
inequivalent SDSs of period $p^n$, over $\Z_p$, is $p^{p^{n-1}}/p^n=p^{p^{n-1} -n}$.
\end{proof}

We continue by demonstrating the ideas for $p=3$.
The correctness of the following lemmas is observed from the definitions.
\begin{lemma}
\label{lem:simpleR}
Let $[X ~~ X+\bfone ~~ X+\bftwo]$ be an SDS of period $3^n$ over $\Z_3$, where $X$ is a word of length $3^{n-1}$.
If $Z$ and $Y$ are two words of length $3^{n-1}$, then the sequence $[V , V+\bfone , V+\bftwo]$, where $V=(Z , Z+Y , Z+2Y+X)$,
is an SDS of period $3^{n+1}$.
\end{lemma}
\begin{proof}
By Definition~\ref{def:SDSZm} the sequence $[V , V+\bfone , V+\bftwo]$ is an SDS and by Lemma~\ref{lem:periodSDS}
its period is~$3^{n+1}$.
\end{proof}

Let $[X ~~ X+\bfone ~~ X+\bftwo]$ be an SDS of period $3^n$ over $\Z_3$, where $X$ is a word of length $3^{n-1}$ over $\Z_3$.
Let $Z$ be a word of length $3^{n-1}$, that starts with a \emph{zero}, over $\Z_3$, and let $Y$ be a word of length $3^{n-1}$ over $\Z_3$.
\begin{lemma}
\label{lem:SDS_Z3}
The sequence $[V , V+\bfone , V+\bftwo]$, where $V=(Z , Z+Y , Z+2Y+X)$,
is an SDS of period $3^{n+1}$ over $\Z_3$. Each different choice of $Z$ and $Y$ yields a distinct SDS.
Each SDS of period $3^{n+1}$ over $\Z_3$ can be constructed in this way.
\end{lemma}
\begin{proof}
By Definition~\ref{def:SDSZm},
a sequence of the form $[V,V+\bfone,V+\bftwo]$, where  $V=(Z , Z+Y , Z+2Y+X)$ and $Z$ is of length $3^{n-1}$, is an SDS
and hence each SDS $[X , X+\bfone , X+\bftwo]$ can be used for such a construction.
Let $S(Z,Y,X)$ denote the sequence
$$
[ Z , Z+Y , Z+2Y+X, Z+\bfone , Z+Y+\bfone , Z+2Y+X+\bfone,Z+\bftwo , Z+Y+\bftwo , Z+2Y+X+\bftwo ]
$$


It is required to show that each SDS of period $3^{n+1}$ is constructed exactly once in this way,
with only one $n$-tuple is taken as $X$ for each SDS of period $3^n$.
First we apply $(\bE^{3^{n-1}}-1)^2= \bE^{2 \cdot 3^{n-1}}-2 \cdot \bE^{3^{n-1}} +1$ on
$[V,V+\bfone,V+\bftwo]$, where  $V=(Z , Z+Y , Z+Y+X)$ and $Z$ is of length $3^{n-1}$. This is done in two steps.
$$
(\bE^{3^{n-1}} -1) [ Z , Z+Y , Z+2Y+X, Z+\bfone , Z+Y+\bfone , Z+2Y+X+\bfone,Z+\bftwo , Z+Y+\bftwo , Z+2Y+X+\bftwo ]
$$
$$
=[Y ,Y+X, Y+2X+\bfone,Y,Y+X, Y+2X+\bfone,Y,Y+X, Y+2X+\bfone].
$$
We continue and observe that
$$
(\bE^{3^{n-1}} -1)[Y ,Y+X, Y+2X+\bfone]=[X , X+\bfone , X+\bftwo].
$$
This implies that for two distinct SDSs of period $3^n$, $[X_1 , X_1+\bfone , X_1+\bftwo]$ and $[X_2 , X_2+\bfone , X_2+\bftwo]$,
$S(Z_1,Y_1,X_1)$ and $S(Z_2,Y_2,X_2)$ are distinct SDSs of period $3^{n+1}$,
where $Z_1,Y_1,Z_2,Y_2$ are four words of length $3^{n-1}$ and furthermore, $Z_1$ and $Z_2$ start with a \emph{zero}.
Moreover, if $Z_1 \neq Z_2$ or $Y_1 \neq Y_2$, then $S(Z_1,Y_1,X)$ is not equivalent to $S(Z_2,Y_2,X)$ as if they were then
$$
[ Z_1 , Z_1+Y_1 , Z_1+2Y_1+X, Z_1+\bfone , Z_1+Y_1+\bfone , Z_1+2Y_1+X+\bfone,Z_1+\bftwo , Z_1+Y_1+\bftwo , Z_1+2Y_1+X+\bftwo ]
$$
$$
= \bE^k [ Z_2 , Z_2+Y_2 , Z_2+2Y_2+X, Z_2+\bfone , Z_2+Y_2+\bfone , Z_2+2Y_2+X+\bfone,Z_2+\bftwo , Z_2+Y_2+\bftwo , Z_2+2Y_2+X+\bftwo ],
$$
where $k \notin \{3^n , 2\cdot 3^n \}$. Applying $(\bE^{3^{n-1}}-1)^2$ on these two sequence implies that
$$
[X , X+\bfone , X+\bftwo] = \bE^\kappa [X , X+\bfone , X+\bftwo],
$$
where $\kappa \equiv k~(\mmod~ 3^n)$. Therefore, the SDS $[X , X+\bfone , X+\bftwo]$ is periodic, a contradiction.

After we have proved that this method yields distinct SDSs of period $3^{n+1}$,
the last step is to prove that all the SDSs of period $3^{n+1}$ were constructed
The proof will be given by a simple enumeration.
The number of possible $Y$'s is $3^{3^{n-1}}$ and the number of
possible $Z$'s is $3^{3^{n-1}-1}$. This implies that from one SDS of period $3^n$ we construct $3^{3^{n-1}-1} \cdot 3^{3^{n-1}}=3^{2 \cdot 3^{n-1}-1}$
SDSs of period $3^{n+1}$ and since the number of SDSs of period $3^n$ is $3^{3^n -n}$, it follows that in total
$3^{3^n -n} \cdot 3^{2 \cdot 3^{n-1}-1}= 3^{3^{n+1} -n-1}$ distinct SDSs of period $3^{n+1}$ were constructed.
Since by Lemma~\ref{lem:enumerateSDS} this is the number of distinct SDSs of period $3^{n+1}$ the proof is completed.
\end{proof}

The same technique and expressions presented in Lemma~\ref{lem:SDS_Z3} for $\Z_3$, can be applied for
any alphabet of a prime size $p$. For example in $\Z_5$, we start with a
SDS $S=[X , X+\bfone , X+\bftwo , X+\bfthree , X+\bffour]$, where $X$ is a word of length $5^{n-1}$.
Let $Y_1$, $Y_2$, and $Y_3$ be words of length $5^{n-1}$ over~$\Z_5$ and let $Z$ be a word of length $5^{n-1}$ that starts with a \emph{zero}.
From $S$, $Y_1$, $Y_2$. $Y_3$, and $Z$, we form the SDS
$[V , V+\bfone , V+\bftwo , V+\bfthree , V+\bffour]$, where
\begin{equation}
\label{eq:SD_binom}
V = (Z , Z+Y_1 , Z+2Y_1 + Y_2 , Z + 3Y_1 +3Y_2 + Y_3, Z+ 4Y_1 +Y_2 +4 Y_3 +X)
\end{equation}
is a word of length $5^n$ over $\Z_5$. On this word we apply $(\bE^{3^{n-1}}-1)^4$ and obtain
the SDS $S=[X , X+\bfone , X+\bftwo , X+\bfthree , X+\bffour]$ and complete as in the proof of Lemma~\ref{lem:SDS_Z3}.
This is further generalized in the following general theorem.

\begin{theorem}
Let $[X , X+\bfone , X+ \bftwo, \ldots, X + \bfone (p-1)]$ be an SDS of period $p^n$ over $\Z_p$, $p$ prime.
Let $Z$ be a word of length $p^{n-1}$ that starts with a \emph{zero} over $\Z_p$ and let
$Y_1, Y_2,\ldots, Y_{p-2}$ be words of length $p^{n-1}$, over $\Z_p$.
\begin{itemize}
\item The sequence $[V , V+\bfone , V+ \bftwo, \ldots, V + \bfone (p-1)]$, where
$$
V=(L_0=Z, L_1=Z+Y_1,L_2 =Z+2Y_1 + Y_2,  \ldots, L_{p-1}),
$$
$$
L_k =Z +\sum_{i=1}^k a_i Y_i, ~~ a_i \equiv \binom{k}{i} ~(\mmod~p), ~~ 1 \leq k \leq p-2
$$
$$
L_{p-1} =Z +\sum_{i=1}^{p-2} a_i Y_i +X, ~~ a_i \equiv \binom{p-1}{i} ~(\mmod~p),
$$
is an SDS of period $p^{n+1}$.

\item Each different choice of $Z$, $Y_1,Y_2,\ldots,Y_{p-1}$ yields a different SDS.

\item Each SDS of period $p^{n+1}$ over $\Z_p$ is constructed in this way.
\end{itemize}
\end{theorem}
\begin{proof}
The proof is essentially the same as the one of Lemma~\ref{lem:SDS_Z3}.
The main difference, except for the increasing number of $Y_i$'s is that instead of applying $(\bE^{3^{n-1}}-1)^2$
on the sequence $[V , V+\bfone , V+ \bftwo]$ we apply $(\bE^{p^{n-1}}-1)^{p-1}$ on the sequence
$[V , V+\bfone , V+ \bftwo, \ldots, V + \bfone (p-1)]$, where $V=(L_0=Z, L_1=Z+Y_1,L_2 =Z+2Y_1 + Y_2, L_3,  \ldots, L_{p-1})$.
Instead of doing it in two steps as done in the proof of Lemma~\ref{lem:SDS_Z3}, it is done in one step using the equation
$$
(\bE^{p^{n-1}}-1)^{p-1} = \sum_{i=0}^{p-1} (-1)^i \binom{p-i}{i} (\bE^{p^{n-1}})^{p-1-i} .
$$
More details on this proof will be discussed in the draft on SDSs~\cite{Etz27}.
\end{proof}

\section{Non-Binary full-period STGCs}
\label{sec:STGC_basics}

In this section, we present and discuss the ordering of non-binary SDSs for the construction of full-period non-binary STGCs.
The requirement is that every two consecutive codewords differ in exactly one coordinate.
The section is partitioned into four subsections.
In Section~\ref{sec:require}, Theorem~\ref{thm:self_dualST} is generalized for a construction of a non-binary alphabet
based on non-binary SDSs. In Section~\ref{sec:differences} we describe a method for constructing
the $m^{m-2}$ SDSs of period $m^2$ for an STGC of length $m$ and period $m^m$ over $\Z_m$. Section~\ref{sec:PropP1P2}
outlines the conditions for an ordering of the SDSs for a recursive construction of the arrangement
required by the generalization of Theorem~\ref{thm:self_dualST}. Section~\ref{sec:recursive} describes the construction for
the arrangement of the SDSs to form the STGC over an alphabet of an odd prime size.

\subsection{Requirements for a construction of STGCs}
\label{sec:require}

Theorem~\ref{thm:single_necklaces} is readily generalized for sequences over $\Z_m$ without any modification.
Theorem~\ref{thm:self_dualST} is generalized to $\Z_m$ with a simple modification as follows.

\begin{theorem}
\label{thm:self_dualST_NB}
Let $S_0, S_1,\ldots,S_{r-1}$ be $r$ binary self-dual pairwise inequivalent full-order SDSs
of period $mn$ over $\Z_m$.
If for each $i$, $0 \leq i < r-1$, $S_i$ and $S_{i+1}$ differ in exactly
$m$ coordinates, and there also exists an integer $\ell$, where $\gcd(\ell,mn)=1$,
such that  $S_{r-1}$ and $\bE^{\ell} S_0$ differ in exactly $m$ coordinates,
then the following words form a length $n$, period $mnr$ STGC,
$$
\begin{array}{lclcccl}
\bF^0 S_0 & \hspace{0.08cm} & \bF^0 S_1 & \hspace{0.08cm} & \cdots & \hspace{0.08cm} & \bF^0 S_{r-1} \\
\bF^\ell S_0 & \hspace{0.08cm} & \bF^\ell S_1 & \hspace{0.08cm} & \cdots & \hspace{0.08cm} & \bF^\ell S_{r-1} \\
\bF^{2\ell} S_0 & \hspace{0.08cm} & \bF^{2\ell} S_1 & \hspace{0.08cm} & \cdots & \hspace{0.08cm} & \bF^{2\ell} S_{r-1} \\
\vdots & \hspace{0.08cm} & \vdots  & \hspace{0.08cm} & \vdots & \hspace{0.08cm} & \vdots \\
\bF^{(mn-1)\ell} S_0 & \hspace{0.08cm} & \bF^{(mn-1)\ell} S_1 & \hspace{0.08cm} & \cdots & \hspace{0.08cm} & \bF^{(mn-1)\ell} S_{r-1} \\
\end{array}.
$$
\end{theorem}

The arrangement of $S_0, S_1,\ldots,S_{r-1}$ defined in Theorem~\ref{thm:self_dualST_NB} will be called the {\bf \emph{initial arrangement}}.
Such an arrangement guarantees the existence of a length $n$, period $mnr$ STGC, over~$\Z_m$.

\begin{example}
\label{exm:STGCm=3,n=3}
For $m=3$ and $n=3$, there are three SDSs of period 9 which contain each
word of length three over $\Z_3$ exactly once. These three sequences are the first three columns on the left of the $9 \times 27$ array,
These three columns form the initial arrangement for these parameters.
The $27$ columns in the $27 \times 9$ array, which follows these three columns, have all the required shifts
of the first three sequences. The first $3$ rows in the array are a length 3 period 27 STGC over $\Z_3$.
$$
\begin{array}{lcccc}
 &  001 & & & 001122222~~112200000~~220011111 \\
 &  011 & & & 011111001~~122222112~~200000220 \\
 &  000 & & & 000220011~~111001122~~222112200 \\
 \hline
 &  112 & & & 112200000~~220011111~~001122222 \\
 &  122 & & & 122222112~~200000220~~011111001 \\
 &  111 & & & 111001122~~222112200~~000220011 \\
 \hline
 &  220 & & & 220011111~~001122222~~112200000 \\
 &  200 & & & 200000220~~011111001~~122222112 \\
 &  222 & & & 222112200~~000220011~~111001122 \\
\end{array}.
$$
%
\hfill\quad $\blacksquare $
\end{example}

\subsection{An initial arrangement based on differences}
\label{sec:differences}

We consider now a construction of the $m^{m-2}$ SDSs of length $m^2$ over $\Z_m$,
which might be used to construct an initial arrangement to form an STGC of length $m$ and period $m^m$.
Such an initial arrangement can be further modified to be used as a seed for the recursive construction
for an STGC of length $p^t$ with period $p^{p^t}$, where $p$ is an odd prime.
We can order the SDSs of period $m^2$ lexicographically such that the first $m$ digits will start with a \emph{zero}
and the sum of these $m$ digits will be 0~modulo~$m$. This simple ordering contains $m^{m-2}$ SDSs of period~$m^2$. Each word
of length~$m$ over~$\Z_m$ is a subsequence of exactly one of these sequences. Now, we should form
an initial arrangement for these SDSs and use Theorem~\ref{thm:self_dualST_NB} to form an STGC.

However, there might be a better way to form these SDSs and arrange them in a simpler way.
Let $\cT$ be the ordered list of the $m^{m-2}$ SDSs of period $m^2$ in an STGC of length $m$ and period $m^m$.
The first $m$ digits (and clearly the first $m+1$ digits) determine the whole SDS.
Assume that the first such $m+1$ digits in a word of an SDS are $0x_2 x_3 ~ \ldots ~ x_m 1$.
There are $m$ such words starting with \emph{zero} in an SDS of period $m^2$.
Consider now the differences $(x_2,x_3 - x_2 , x_4 -x_3,\ldots, x_m - x_{m-1}, 1 - x_m)$,
i.e., $(\bE -1) (0,x_2,x_3,\cdot,x_m,1)$, in this word.
Each word of length $m+1$ that starts with a \emph{zero} defines a unique such sequence of $m$ differences.
The sum of these differences is 1 modulo $m$.
We order this set of differences lexicographically.
In this order, there are $m^{m-1}$ sequences of differences of length~$m$, where the sum of the differences in each sequence is 1~modulo~$m$.
Each SDS has $m$ such sequences of differences. Every two such sequences are equivalent as cyclic sequences, but different as words of length $m$.

\begin{example}
For $n=m=4$, the 64 sequences of differences are ordered lexicographically in a $4 \times 64$ array as follows.
$$
\begin{array}{lc}
 &  0000000000000000111111111111111122222222222222223333333333333333 \\
 &  0000111122223333000011112222333300001111222233330000111122223333 \\
 &  0123012301230123012301230123012301230123012301230123012301230123 \\
 &  1032032132102103032132102103103232102103103203212103103203213210 \\
\end{array}.
$$
\hfill\quad $\blacksquare $
\end{example}

Two sequences of differences are associated with the same SDS if and only if one is a cyclic shift of the other.
From all possible $m$ cyclic shifts, only one sequence of differences is taken as a representative.
For each sequence of differences, there are $m$ distinct cyclic shifts since
the sum 1~modulo~$m$ of the differences implies that there is no periodicity in a sequence of differences.
Therefore, there are $m^{m-2}$ cyclic sequences of differences in the final set of differences.
This set of sequences (the representatives) containing the differences is ordered in a way that every two consecutive sequences
(including the last and the first)
have an identical subsequence of length $m-2$. Therefore, in the corresponding SDSs of length $m^2$,
the $m^2$~subsequences of length $m$ can be paired such that each pair (including the last and the first in
an appropriate shift) differs in exactly one coordinate.
Thus, the conditions of Theorem~\ref{thm:self_dualST_NB} are satisfied and an STGC of length $m$ and period $m^m$ is constructed.

\begin{example}
\label{ex:m=4}
For $n=m=4$, the 64 sequences of differences are ordered in groups of size 4, where all the sequences in a group are
cyclic shifts of one another and ordered lexicographically. They are ordered in a $4 \times 64$ array, where
the groups are ordered in the same order as the order of the final $16$ representatives.
$$
\begin{array}{lc}
 &  0001~0122~1112~0122~0023~0333~1233~1233~2223~0122~0113~1233~0113~0133~0023~0023 \\
 &  0010~1202~1121~2201~2300~3033~3123~3312~2232~2012~1031~2313~3011~1303~0230~0302 \\
 &  0100~2210~1211~1022~0032~3303~3312~2133~2322~2201~3110~3321~1301~3310~2300~3200 \\
 &  1000~2021~2111~2210~3200~3330~2331~3321~3222~1220~1301~3132~1130~3031~3002~2030 \\
\end{array}.
$$

One representative sequence of differences is taken from each group of $4$.
The 16 representative sequences of differences are ordered in a $4 \times 16$ array such that every
two consecutive columns have an identical subsequence of length~2.
$$
\begin{array}{lc}
 &  0012232331111100 \\
 &  1111001120021123 \\
 &  0222333322130332 \\
 &  0210033222333000 \\
\end{array}.
$$
Each sequence of differences is translated to a word of length~5, which is extended to an SDS of length~16.
These SDSs are the columns of the following $16 \times 16$ array. The 16 sequences defined by the $16$ columns form
the arrangement of the 16 SDSs to satisfy the conditions of Theorem~\ref{thm:self_dualST_NB}.
This arrangement leads to an STGC of length 4 and period 256.
$$
\begin{array}{ccc}
 0032211000000000 & 0 & \cdots \\
 0000003331111100 & 0 & \cdots \\
 1111000011132223 & 0 & \cdots \\
 1333333333222111 & 1 & \cdots \\
 \hline
 1103322111111111 & 1 & \cdots \\
 1111110002222211 & 1 & \cdots \\
 2222111122203330 & 1 & \cdots \\
 2000000000333222 & 2 & \cdots \\
 \hline
 2210033222222222 & 2 & \cdots \\
 2222221113333322 & 2 & \cdots \\
 3333222233310001 & 2 & \cdots \\
 3111111111000333 & 3 & \cdots \\
 \hline
 3321100333333333 & 3 & \cdots \\
 3333332220000033 & 3 & \cdots \\
 0000333300021112 & 3 & \cdots \\
 0222222222111000 & 0 & \cdots \\
\end{array}.
$$
\hfill\quad $\blacksquare $
\end{example}

\subsection{Further requirements for the recursive construction}
\label{sec:PropP1P2}

The next step is to construct a full-period STGC for words whose length is $m^t$, where $t > 1$.
The construction will be recursive, where usually the basis is the construction for words of length~$m$, i.e., $t=1$,
derived from $m^{m-2}$ SDSs whose period is $m^2$ ($m=3$ requires a larger basis with $729$ SDSs of period $27$). Such a recursive construction
will be given when $m=p$ is an odd prime, and such a basis was presented for $p=3$ and $p=5$.

The construction will start with an initial arrangement of $p^{p^{t-1}-t}$ SDSs of period $p^t$, $t > 1$,
as described in Theorem~\ref{thm:self_dualST_NB}. By Theorem~\ref{thm:self_dualST_NB}, this initial arrangement implies the existence
of an STGC of length $p^{t-1}$ and period $p^{p^{t-1}}$. For the recursive construction, we start with
such an initial arrangement $S_0,S_1,\ldots,S_{r -1}$, where $r = p^{p^{t-1}-t}$.
From this initial arrangement, two more properties are required.

\begin{itemize}
\item[(P1)] For each $i$, $0 \leq i \leq r-1$, $S_i$ and $S_{i+1}$ (subscripts taken modulo $r$)
differ in exactly $p$ positions, $k+ j \cdot p^{t-1}$, for each $0 \leq j \leq p-1$, where $0 \leq k \leq p^{t-1} -1$
(note that in Theorem~\ref{thm:self_dualST_NB} there is no such requirement for $i=r-1$).

\item[(P2)] Let diff$(S_i,S_{i+1})$ denote the first position in which $S_i$ and $S_{i+1}$ differ and let
$$
\cD_{p^t} = \{\text{diff}(S_i,S_{i+1}) ~:~ 0 \leq i \leq r-1 \} .
$$
Then $\cD_{p^t} = \{0,1,\ldots,p^{t-1}-1\}$.
\end{itemize}

The recursive construction has several steps, and it is a modification for the construction of binary STGC
of length $2^t$ and period $2^{2^t} - 2^{t+1}$ as was described in~\cite{EtPa96}.

Before starting to describe the general method for any odd prime $p$ we consider $p=3$.
The simpler recursive construction for binary SDSs was considered first in~\cite{EtLe84}.

The basis of the recursive construction is the initial arrangement of $729$ SDSs of length $27$ (given in Appendix~\ref{sec:seed_p3}).
This arrangement for these $729$ SDSs of period $27$ satisfies properties (P1) and (P2).

\subsection{The recursive construction}
\label{sec:recursive}

Now, we assume the existence of such an initial arrangement of the $3^{3^{t-1}-t}$ SDSs of period~$3^t$, $t >2$,
for an STGC of length $3^{t-1}$ and period $3^{3^{t-1}}$ that also satisfies properties (P1) and (P2).
Each SDS in such an initial arrangement has period $3^t$ and the form $[X , X+\bfone , X+\bftwo]$, where $X$ is
a word of length $3^{t-1}$.

We now have to describe such an initial arrangement of the $3^{3^t-t-1}$ SDSs of period $3^{t+1}$, $t > 1$.
Each SDS in this arrangement has period $3^{t+1}$ and the form $[V , V+\bfone , V+\bftwo]$. It is constructed
from an SDS $[X , X+\bfone , X+\bftwo]$, where $X$ is a word of length $3^{t-1}$ and
$V=(Z , Z+Y , Z+2Y+X)$, as proved in Lemma~\ref{lem:SDS_Z3}. The words $Z$ and $Y$ are of length $3^{t-1}$, where $Z$ can be any word
over~$\Z_3$ that starts with a \emph{zero} and $Y$ can be any word of length $3^{t-1}$ over $\Z_3$.

When $Z$ and $Y$ are fixed, then for $X$ we are using the word $X$ from $[X , X+\bfone , X+\bftwo]$ of all the
SDSs of the initial arrangement (having also properties (P1) and (P2)~) with SDSs of period $3^t$ to form SDSs of period $3^{t+1}$ using the recursion
$[V , V+\bfone , V+\bftwo]$, where $V=(Z , Z+Y , Z+2Y+X)$.
The outcome is an arrangement (not initial) with $3^{3^{t-1}-t}$ SDSs of period $3^{t+1}$.
In this arrangement every two consecutive SDSs differ in just one position on the first $3^t$ coordinates.
Moreover, by (P1) this arrangement is cyclic as the first SDS and the last SDS also differ in just one position on the first $3^t$ coordinates.
This component (arrangement) will be denoted by $A(Z,Y)$.
Since $X$ influences only the last $3^{t-1}$ positions of $V$, it follows that for two consecutive sequences in this arrangement
$S_i$ and $S_{i+1}$ we have $2 \cdot 3^{t-1} \leq \text{diff}(S_i,S_{i+1}) \leq 3^t -1$ and since $\cD_{3^t} = \{0,1,\ldots,3^{t-1}-1\}$,
each value in the range $\{ 2 \cdot 3^{t-1}, 2 \cdot 3^{t-1}+1,\ldots, 3^t-1\}$ is attained for $\cD_{3^{t+1}}$.
There are $3^{3^{t-1}-1}$ words of length $3^{t-1}$ that start with a \emph{zero} and $3^{3^{t-1}}$ words of length $3^{t-1}$.
Therefore, the number of distinct components constructed in this way is $3^{3^{t-1}-1}\cdot 3^{3^{t-1}}=3^{3^t - 3^{t-1} -1}$.

%

The next stage is more complicated, and in this step, we merge all the $3^{3^t - 3^{t-1} -1}$ components into
an initial arrangement with $3^{3^t -t-1} $ SDSs of period $3^{t+1}$ that satisfies properties (P1) and (P2).
We start by merging all the components whose SDSs share the same values in the first $3^{t-1}$ coordinates.
We partition all the nonzero ternary words of length $3^{t-1}$
into $\frac{3^{3^{t-1}-1}}{2}$ pairs, such that the two words in each pair
have the same weight and they differ in exactly one nonzero coordinate.
These pairs are ordered in a list $\cL$ by their weights from smallest to the largest weight,
where the order between pairs of the same weight is arbitrary. For a given word $Z$, of length $3^{t-1}$, that starts with a \emph{zero},
at each step now we will have an arrangement $\cA(Z)$ obtained by merging some of the components of the form $A(Z,Y)$.
The arrangement starts with the component $A(Z,\bfzero)$. At the general step we will have an arrangement $\cA(Z)$ obtained
so far and the first pair $\{ Y_1,Y_2 \}$ in the list $\cL$ whose components $A(Z,Y_1)$ and $A(Z,Y_2)$ were not merged into $\cA(Z)$.
The two words $Y_1$ and $Y_2$ differ only in one coordinate, the $k$-th coordinate. Let $S=[X_1,X_1+\bfone,X_1+\bftwo]$ and
$S'=[X_2,X_2+\bfone,X_2+\bftwo]$ be two consecutive SDSs of the initial arrangement that differ in the $k$-th coordinate,
where both have a different nonzero element. Let $Y_3$
be the word of length $3^{t-1}$ that differ from $Y_1$ and $Y_2$ only in the $k$-th coordinate. It has weight less by one
than $Y_1$ (also $Y_2$) and hence the
component $A(Z,Y_3)$ was already merged into $\cA(Z)$. We consider now the following two consecutive SDSs in the three components
$A(Z,Y_1)$, $A(Z,Y_2)$ and $\cA(Z)$ (the last was before in $A(Z,Y_3)$).

\noindent
In $A(Z,Y_1)$ we consider the following SDSs denoted by $(Y_1,X_1)$ and $(Y_1,X_2)$:
$$
[Z,Z+Y_1,Z+2Y_1 +X_1,Z+\bfone,Z+Y_1+\bfone,Z+2Y_1 +X_1+\bfone,Z+\bftwo,Z+Y_1+\bftwo,Z+2Y_1 +X_1+\bftwo]
$$
$$
[Z,Z+Y_1,Z+2Y_1 +X_2,Z+\bfone,Z+Y_1+\bfone,Z+2Y_1 +X_2+\bfone,Z+\bftwo,Z+Y_1+\bftwo,Z+2Y_1 +X_2+\bftwo]
$$

\noindent
In $A(Z,Y_2)$ we consider the SDSs $(Y_2,X_1)$ and $(Y_2,X_2)$:
$$
[Z,Z+Y_2,Z+2Y_2 +X_1,Z+\bfone,Z+Y_2+\bfone,Z+2Y_2 +X_1+\bfone,Z+\bftwo,Z+Y_2+\bftwo,Z+2Y_2 +X_1+\bftwo]
$$
$$
[Z,Z+Y_2,Z+2Y_2 +X_2,Z+\bfone,Z+Y_2+\bfone,Z+2Y_2 +X_2+\bfone,Z+\bftwo,Z+Y_2+\bftwo,Z+2Y_2 +X_2+\bftwo]
$$

\noindent
In $\cA(Z)$ we consider the SDSs $(Y_3,X_1)$ and $(Y_3,X_2)$ (that were before in $A(Z,Y_3)$):
$$
[Z,Z+Y_3,Z+2Y_3 +X_1,Z+\bfone,Z+Y_3+\bfone,Z+2Y_3 +X_1+\bfone,Z+\bftwo,Z+Y_3+\bftwo,Z+2Y_3 +X_1+\bftwo]
$$
$$
[Z,Z+Y_3,Z+2Y_3 +X_2,Z+\bfone,Z+Y_3+\bfone,Z+2Y_3 +X_2+\bfone,Z+\bftwo,Z+Y_3+\bftwo,Z+2Y_3 +X_2+\bftwo]
$$

The goal is to merge the three lists $A(Z,Y_1)$, $A(Z,Y_2)$ and $\cA(Z)$, where these SDSs
in a new arrangement will be in the following order
\begin{equation}
\label{eq:mergeY1Y2Y3}
\ldots, (Y_3,X_2), (Y_2,X_1), \ldots, (Y_2,X_2), (Y_1,X_1), \ldots,   (Y_1,X_2), (Y_3,X_1), \ldots ~~~.
\end{equation}
To fulfill this goal, it is required that the following equalities will be satisfied:
$$
Z + 2Y_3 +X_2 = Z+2Y_2 +X_1
$$
$$
Z + 2Y_2 +X_2 = Z+2Y_1 +X_1
$$
$$
Z + 2Y_1 +X_2 = Z+2Y_3 +X_1 ~.
$$
These three equalities are equivalent to the following equalities:
$$
Y_3 - Y_2 = Y_2 -Y_1 = Y_1 - Y_3
$$
and
$$
X_1 - X_2 = 2(Y_3 - Y_2) = 2(Y_2 - Y_1) = 2(Y_1 - Y_3)~.
$$

Since $X_2$ and $X_1$ differs on the in the $k$-th coordinate and the same is true for each pair of $Y_1,Y_2,Y_3$ we
can consider only one coordinate (the $k$-th coordinate) of these words. Let $y_i$ denote the $k$-th coordinate of $Y_i$
and $x_i$ the $k$-th coordinate of $X_i$. Since the order of the SDSs in an arrangement can be reversed and we can also
switch between $Y_1$ and $Y_2$, it follows that
w.l.o.g. we can assume that $x_1 =0$ and $x_2 =2$. Furthermore, since $\{ y_1,y_2,y_3\} = \{0,1,2\}$, it follows
that $Y_1$ can be chosen in a way that $2=x_2 - x_1 = 2(y_3 -y_1)$, i.e., $y_3 = y_1 +1$ and $y_2 = y_3 +2$.
Hence, we can merge $A(Z,Y_1)$ and $A(Z,Y_2)$ into $\cA(Z)$ to obtain a new arrangement $\cA(Z)$.
The process ends when each component $A(Y,Z)$ is merged into $\cA(Z)$. Moreover, the $k$-th coordinate can be chosen
in a way that each integer in $\{0,1,\ldots,3^{t-1}-1\}$ is chosen at least once and hence each value
in the range $\{ 3^{t-1}, 3^{t-1}+1,\ldots, 2 \cdot 3^{t-1}-1\}$, is attained for $\cD_{3^{t+1}}$.

We are now in a position to merge all arrangements $\cA(Z)$, where $Z$ is any word of length $3^{t-1}$ that starts with a \emph{zero}
and obtain $\cD_{3^{t+1}} = \{ 0,1,2,\ldots,3^t -1 \}$.
We continue and at a general step we have an arrangement $\cA$ obtained so far and the list of words
of length $3^{n-1}$ that start with a \emph{zero} ordered by weights.
At the general step we will have an arrangement $\cA$ obtained
so far and the first element $Z_1$ in the list whose components $\cA(Z_1)$ were not merged into $\cA$.
We continue similarly to the mwthod used before to merge $\cA(Z_1)$ into $\cA$.
Let $Z_2$ be a word of length $3^{n-1}$ that starts with a \emph{zero} and differs in exactly one coordinate, say $k$, from $Z_1$.
Let $Y_1$ and $Y_2$ be two words of length $3^{t-1}$ that differ only in the $k$-th coordinate and $Y_2 - Y_1 = Z_1 -Z_2$
(note that there are many possible choices for $Y_1$ and $Y_2$).
Let further $X_1$ and $X_2$ be two words of length $3^{t-1}$ that differ in the $k$-th coordinate, $X_2 - X_1 = Y_1-Y_2$,
and are contained in the two SDSs $[X_1 ,X_1 + \bfone ,X_1 +\bftwo]$ and $[X_2 ,X_2 + \bfone ,X_2 +\bftwo]$ of the initial
arrangement for SDSs of period $3^t$. Given $Z_1$ and $Z_2$, we choose $Y_1$ and $Y_2$ such that $Y_2 - Y_1 = Z_1 -Z_2$
and also $X_1$ and $X_2$ that differ in the $k$-the coordinate such that $X_2 - X_1 = Y_1-Y_2$ and are the first $3^{t-1}$ digits
in the two SDSs differing in this coordinate.
By the construction so far we have in $A(Z_1,Y_1)$ the following two SDSs,
$(Z_1,Y_1,X_1)$ and $(Z_1,Y_1,X_2)$, that are consecutive SDSs in this component:
$$
[Z_1,Z_1+Y_1,Z_1+2Y_1 +X_1,Z_1+\bfone,Z_1+Y_1+\bfone,Z_1+2Y_1 +X_1+\bfone,Z_1+\bftwo,Z_1+Y_1+\bftwo,Z_1+2Y_1 +X_1+\bftwo]
$$
and
$$
[Z_1,Z_1+Y_1,Z_1+2Y_1 +X_2,Z_1+\bfone,Z_1+Y_1+\bfone,Z_1+2Y_1 +X_2+\bfone,Z_1+\bftwo,Z_1+Y_1+\bftwo,Z_1+2Y_1 +X_2+\bftwo].
$$
In $A(Z_2,Y_2)$ we have the following two SDSs that are consecutive SDSs, and denoted by
$(Z_2,Y_2,X_1)$ and $(Z_2,Y_2,X_2)$, in this component:
$$
[Z_2,Z_2+Y_2,Z_2+2Y_2 +X_1,Z_2+\bfone,Z_2+Y_2+\bfone,Z_2+2Y_2 +X_1+\bfone,Z_2+\bftwo,Z_2+Y_2+\bftwo,Z_2+2Y_2 +X_1+\bftwo]
$$
and
$$
[Z_2,Z_2+Y_2,Z_2+2Y_2 +X_2,Z_2+\bfone,Z_2+Y_2+\bfone,Z_2+2Y_2 +X_2+\bfone,Z_2+\bftwo,Z_2+Y_2+\bftwo,Z_2+2Y_2 +X_2+\bftwo].
$$
Since $Y_2 - Y_1 = Z_1 -Z_2$. it follows that $Z_2 + Y_2 = Y_1 + Z_1$ and since also $X_2 - X_1 = Y_1-Y_2$, it follows that
$Z_1 + 2Y_1 + X_1 = Z_2 + 2Y_2 + X_2$. Hence, we can form a new arrangement in $\cA$, with the following order of these four SDSs of period $3^{t+1}$.
$$
\ldots, (Z_1,Y_1,X_1),(Z_2,Y_2,X_2), \ldots, (Z_2,Y_2,X_1),(Z_1,Y_1,X_2),\ldots
$$
This also add the integer $k$ to $\cD_{3^{t+1}}$.
This process yields all the values in the set $\{1,2,\ldots, 3^{t-1}  -1 \}$ to $\cD_{3^{t+1}}$.
However, the value $0$ is missing in $\cD_{3^{t+1}}$, so we need to obtain it in a slightly different way.

Recall that if $[X ,X + \bfone ,X +\bftwo]$ is an SDS of period $3^n$, then $[V ,V + \bfone ,V +\bftwo]$,
where $V=(Z , Z+Y, Z+2Y+X)$, and $Z$ and $Y$ are words $3^{n-1}$ is an SDS of period $3^{n+1}$.
If $Z+ \bfone$ or $Z+\bftwo$ are taken instead of $Z$, the same SDS is obtained. Hence, we can use the
the all-one word of length $3^{n-1}$ for $Z$ in the recursive construction of SDSs of period $3^{n+1}$ instead of the all-zero word.
In fact, no change should be done, except that when $Z_1=(011 ~\cdots~ 11)$, we consider $Z_2$ to be the all-one word and
with the same process we obtain the value $0$ for $\cD_{3^{t+1}}$.

Finally, the partition of pairs described in the construction is possible and proved in the following lemma.
\begin{lemma}
The set of nonzero ternary words of length $n$, $n \geq 1$, can be partitioned into pairs such that the two words in a pair have the same weight
and differ in exactly one coordinate.
For each coordinate there exists at least one pair that differs in that coordinate.
\end{lemma}
\begin{proof}
The proof is by induction. The basis for $n=1$ is trivial.

Assume that the claim is true for some $n \geq 1$.

Consider the nonzero ternary words of length $n+1$. We are using the following partition of these words.
If two words $X$ and $Y$ of length $n$ are paired by the induction
hypothesis, then we use the following three pairs for words of length $n+1$,
$\{(X,0),(Y,0)\}$, $\{(X,1),(X,2)\}$, and $\{(Y,1),(Y,2)\}$. This complete the partition into pairs
of all the nonzero ternary words. By this definition, using also the induction hypothesis,
the two words in a pair have the same weight and they differ in exactly one nonzero coordinate.
Moreover, pairs of the form $\{(X,1),(X,2)\}$ differs in the last coordinate, and on each other coordinate.
If $X$ and $Y$ differ in the $k$-th coordinate, then any pair $\{(X,0),(Y,0)\}$ differs in the $k$-th coordinate. Hence,
by the induction hypothesis, for each coordinate there is at least one pair that differs in this coordinate.
\end{proof}

The same procedure will work for any given prime $p$, but merging the components will become more evolved. For example,
when $p=5$ we will have to consider $Z$, $Y_1$, $Y_2$, and $Y_3$.
The  first step will be to generate a component $A(Z,Y_1,Y_2,Y_3)$ for each for words $Y_1$, $Y_2$, $Y_3$ of length $5^{t-1}$ over $\Z_5$
and a word $Z$ of length $5^{t-1}$, that starts with a \emph{zero}, over $\Z_5$.
The second step  will be to merge components with the same $Y_1$, $Y_2$, and $Z$, into
a component $A(Z,Y_1,Y_2)$ and after that to merge these larger components into $A(Z,Y_1)$.
Finally, for each word $Z$ of length $5^{t-1}$ that starts with a \emph{zero} we have a components $\cA(Z)$.
In the process we obtain the required $\cD_{5^{t+1}}$.

Having the right seeds for $p=3$ and $p=5$ (see Appendix~\ref{sec:seed_p3} and Appendix~\ref{sec:seed_p5}, respectively),
an STGC of length $p^t$ and period $p^{p^t}$ was constructed for each $t \geq 1$.

\section{Conclusions and Directions for Future Research}
\label{sec:conclude}

A recursive construction for non-binary STGCs based on non-binary SDSs is presented.
This construction is implemented to construct STGCs of length $p^t$ and period $p^{p^t}$,
when $p=3$ and $p=5$ for $t \geq 2$. These are the first and only known two infinite families of STGCs of full-period.
When $p>5$ is an odd prime, the construction can be implemented to construct such codes if an appropriate
seed for the $p^{p-2}$ SDSs of period $p^2$ is constructed. This new direction of research raised
some old open problems again and some new ones as well as some other directions for future research.

\begin{enumerate}
\item Binary SDSs are generated by a feedback shift register known as the complemented cycling register.
The generalization of this register for non-binary alphabets of prime size and non-prime size raises
many interesting construction and enumeration problems. Some of them are considered in~\cite{Etz26,Etz27}.
However, there are more directions not considered in these papers.

\item It is conjectured that binary STGCs of length $p$ and period $2^p-2$ can be obtained by using
Theorem~\ref{thm:single_necklaces} and also by using Theorem~\ref{thm:self_dualST}.
Such codes are known to exist for all primes up to $p=19$~\cite{Etz07}.
Can the conjecture be proved, at least for an infinite family of primes?

\item Using Theorem~\ref{thm:self_dualST}, other binary STGCs, with maximum possible period, might exist.
Can such STGCs be constructed? Again, we are looking for an infinite family of such codes.

\item The construction given in the paper can be used only for primes. What about a construction
for other alphabet sizes? When $m=4$, a construction for STGC of length $4$ and period $256$ was given in Example~\ref{ex:m=4}.

\item An ordering of the differences was suggested as a potential idea for constructing and ordering SDSs.
These orderings were used to find the ordering of the SDSs when $m=4$ and $m=5$. Can it be done for larger $m$?

\item Can appropriate seeds be found for primes larger than $5$? We believe that using some effort and computer search,
more seeds can be found.
\end{enumerate}

\appendices

\section{Seed arrangement for $p=3$}
\label{sec:seed_p3}

Unfortunately for $p=3$ there is no arrangement of the $p^{p-2}=3$ SDSs of length $9$ that
satisfies the condition of Theorem~\ref{thm:self_dualST_NB} and also Properties (P1) and (P2).
This forced us to look for a seed with the $729$ SDSs of period $27$. There is certainly a way to use
a computer search and to find such an arrangement. However, a more interesting way is to find
a combinatorial method to find such an arrangement. This combinatorial way might help to find seeds
for other odd primes.

There are three SDSs of period $9$ and we arrange them with a fixed 3-tuple, for each one, to be regarded as the first one as follows

$$
\begin{array}{cc}
A & =[000111222] \\
B & =[010121202] \\
C & =[110221002]
\end{array}
$$

This is an initial arrangement that satisfies Theorem~\ref{thm:self_dualST_NB} as demonstrated in Example~\ref{exm:STGCm=3,n=3}.

We are now using another method to construct and arrange the SDSs of period $27$.
This method called {\bf \emph{interleaving}} was used efficiently in multidimensional coding~\cite{BBV98,EtVa02},
construction of covering sequences~\cite{CETV25}, etc. It will be further discussed in detail in~\cite{Etz27}.
Let $S_1,S_2,S_3$ be three SDSs of period $3^n$ (not necessarily distinct) over $\F_3$.
For each sequence, a starting point is selected such that the $i$-th sequence is
$$
S_i = [s_{i,0},s_{i,1},\ldots,s_{i,3^n-1}].
$$
The {\bf \emph{interleaving}} of the sequences $S_1,S_2,S_3$, is the sequence
$$
S = [s_{1,0},s_{2,0},s_{3,0},s_{1,1},s_{2,1},s_{3,1},\ldots,s_{1,3^n-1},s_{2,3^n-1},s_{3,3^n-1}],
$$
where $s_{i,j}$ is the value of coordinate $j$ in the sequence $S_i$. The sequence $S$ is an SDS of period~$3^{n+1}$
and each SDS of period~$3^{n+1}$ can be constructed in this way.
The first sequence $S_1$ is also taken at its starting point, i.e., zero shift.
The other two sequences can be taken in $3^n$ possible shifts. If the three interleaved sequences are not equivalent
(two can be equivalent) then all the constructed SDSs of period $3^{n+1}$ are distinct.

There are $11$ possible cyclic orders of the three SDSs of period $9$ and they are of three types:
\begin{itemize}
\item[(1)] Two sequences participate in the interleaving, i.e., there are size
subtypes $[AAB]$, $[AAC]$, $[BBA]$, $[BBC]$, $[CCA]$, and $[CCB]$.

\item[(2)] All the sequences to be interleaved are different, i.e., there are two subtypes $[ABC]$ and $[ACB]$.

\item[(3)] All the sequences to be interleaved are the same, i.e., there are three subtypes $[AAA]$, $[BBB]$, and $[CCC]$.
\end{itemize}

Let us enumerate the number of sequences of each type.

When exactly two SDSs that are interleaved are the same, e.g., the subtype $[AAB]$, we can order the first appearance and the second appearance
in the interleaving. In other words, we can decide that the first sequence is $A$, the second is $A$, and the third is $B$.
We can decide that the first sequence is $A$, the second is $B$, and the third is $A$.
We can also decide that the first sequence is $B$, the second is $A$, and the third is~$A$.
Assume we decided that the first sequence is $A$, the second is~$B$, and the third is $A$.
We will take the first sequence $A$ in its zero shift, i.e., $[000111222]$, the second $A$ and $B$ will
be taken in all possible $9$ shifts to obtain $81$ SDSs of period $27$.
The same is true for the other subtypes, i.e., $[AAC]$, $[BBA]$, $[BBC]$, $[CCA]$, and $[CCB]$
and hence there are $6$ cyclic orders in this type for a total of $486$ SDSs.
The same arguments work for the second type that contains $[ABC]$ and $[ACB]$.
These add $162$ SDSs of period $27$.

The last type in which all three sequences, which are interleaved, are the same is slightly more complicated.
The first one is taken in its zero shift, the second one in shift $i$, and the third one in shift $j$, $0 \leq i,j \leq 8$.
But when we shift the SDS of length $27$ to the left, we obtain that the new first sequence is at shift $i$, the second
sequence is at shift $j$, and the third sequence is at shift $1$. This is equivalent to the situation that the new first sequence is at the zero
shift, the second sequence is at shift $j-i$, and the third sequence is at shift $1-i$. Another shift to the left we obtain a new first sequence
with shift $j$, the new second sequence will be at shift $1$, and the third one at shift~$2$.
This is equivalent to the situation that the new first sequence is at the zero
shift, the second sequence is at shift $1-j$, and the third sequence is at shift $2-j$.
All these three possible ways for interleaving, where we start with a sequence at its zero shift, yield the same SDS of period $27$.
All these three possible shifts form the same SDS sequence and hence the last two should be omitted.
Thus, we have only $27$ SDSs when we interleave the three $A$'s, $27$ SDSs for three $B$'s and
$27$ SDSs for three $C$'s for a total of $81$ SDSs of period $27$.
This scenario is demonstrated in the following table.

$$
\begin{array}{lccccc}
 &\text{first~shift~vector} & \text{after~a~shift} & \text{second~shift~vector} & \text{another~shift} & \text{third~shift~vector} \\
& [0,0,0] & [0,0,1] & [0,0,1] & [0,1,1] & [0,1,1] \\
& [0,0,2] & [0,2,1] & [0,2,1] & [2,1,1] & [0,8,8] \\
& [0,0,3] & [0,3,1] & [0,3,1] & [3,1,1] & [0,7,7] \\
& [0,0,4] & [0,4,1] & [0,4,1] & [4,1,1] & [0,6,6] \\
& [0,0,5] & [0,5,1] & [0,5,1] & [5,1,1] & [0,5,5] \\
& [0,0,6] & [0,6,1] & [0,6,1] & [6,1,1] & [0,4,4] \\
& [0,0,7] & [0,7,1] & [0,7,1] & [7,1,1] & [0,3,3] \\
& [0,0,8] & [0,8,1] & [0,8,1] & [8,1,1] & [0,2,2] \\
& [0,1,0] & [1,0,1] & [0,8,0] & [8,0,1] & [0,1,2] \\
& [0,1,3] & [1,3,1] & [0,2,0] & [2,0,1] & [0,7,8] \\
& [0,1,4] & [1,4,1] & [0,3,0] & [3,0,1] & [0,6,7] \\
& [0,1,5] & [1,5,1] & [0,4,0] & [4,0,1] & [0,5,6] \\
& [0,1,6] & [1,6,1] & [0,5,0] & [5,0,1] & [0,4,5] \\
& [0,1,7] & [1,7,1] & [0,6,0] & [6,0,1] & [0,3,4] \\
& [0,1,8] & [1,8,1] & [0,7,0] & [7,0,1] & [0,2,3] \\
& [0,2,4] & [2,4,1] & [0,2,8] & [2,8,1] & [0,6,8] \\
& [0,2,5] & [2,5,1] & [0,3,8] & [3,8,1] & [0,5,7] \\
& [0,2,6] & [2,6,1] & [0,4,8] & [4,8,1] & [0,4,6] \\
& [0,2,7] & [2,7,1] & [0,5,8] & [5,8,1] & [0,3,5] \\
& [0,3,2] & [3,2,1] & [0,8,7] & [8,7,1] & [0,8,2] \\
& [0,3,6] & [3,6,1] & [0,3,7] & [3,7,1] & [0,4,7] \\
& [0,4,2] & [4,2,1] & [0,7,6] & [7,6,1] & [0,8,3] \\
& [0,4,3] & [4,3,1] & [0,8,6] & [8,6,1] & [0,7,2] \\
& [0,5,2] & [5,2,1] & [0,6,5] & [6,5,1] & [0,8,4] \\
& [0,5,3] & [5,3,1] & [0,7,5] & [7,5,1] & [0,7,3] \\
& [0,5,4] & [5,4,1] & [0,8,5] & [8,5,1] & [0,6,2] \\
& [0,6,3] & [6,3,1] & [0,6,4] & [6,4,1] & [0,7,4]
\end{array}
$$

Now, we consider the three SDSs of period $9$ again. They have the following {\bf \emph{adjacency}} properties.
For the sequence $A$ we can observe that $A$ differ in exactly $3$ positions from $\bE A$ and from $\bE^{-1} A$.
For the sequence $B$ we can observe that $B$ differ in exactly $3$ positions from $\bE^2 B$ and from $\bE^{-2} B$.
For the sequence $C$ we can observe that $C$ differ in exactly $3$ positions from $\bE^4 C$ and from $\bE^{-4} C$.

Consider, for example, the subtype $[AAB]$. There are $81$ distinct SDSs with this cyclic order associated
with $81$ shift vectors, e.g., all the shift vectors of the form $[0,i,j]$, $0 \leq i,j \leq 8$. These shift vectors
will be ordered so that two consecutive ones (including the last and the first) $[0,i,j]$ and $[0,i',j']$ have one of the following
two properties:
\begin{itemize}
\item $i=i'$ and $\abs{j-j'}=2$.

\item $j=j'$ and $\abs{i-i'}=1$.
\end{itemize}
This arrangement is guaranteed by the adjacency properties to ensure that the two consecutive SDSs of period~$27$ obtained by interleaving
using the corresponding shifts will differ in exactly three positions. This arrangement is quite simple, e.g.,
for a fixed $i$, if the first shift vector in this arrangement is $[0,i,j]$ the next $9$ shift vectors are
$$
\begin{array}{lc}
& [0,i,j+2] \\
& [0,i,j+4] \\
& [0,i,j+6] \\
& [0,i,j+8] \\
& [0,i,j+1] \\
& [0,i,j+3] \\
& [0,i,j+5] \\
& [0,i,j+7] \\
& [0,i+1,j+7] \\
\end{array} .
$$
The coordinate in which all the SDSs are in their zero shift will be called {\bf \emph{pivot}}.
This arrangement can start with any shift vector $[0,i,j]$ and hence it will end with the shift vector $[0,i-1,j]$.
The same ideas work for the other cyclic orders of this type. For $[AAC]$ and two consecutive shift vectors $[0,i,j]$ and $[0,i',j']$
we will have either $\abs{i-i'}=1$ or $\abs{j-j'}=4$. Similar arrangement will be done for $[BBA]$, $[BBC]$, $[CCA]$, and $[CCB]$.
For the second type that contains $[ABC]$ and $[ACB]$, the same idea will work.

For the third type that contains $[AAA]$, $[BBB]$, and $[CCC]$, the same idea does not work since we have only 27 shift
vectors for each subtype with no specific rule to arrange them. Hence, they will be left without an arrangement
with cycles of length $27$ as was done for the other $8$ subtypes.

The next step is to merge the $8$ cycles of the $8$ subtypes. For this step, we consider again the three SDSs of period $9$, $A$, $B$, and $C$.
For the sequence $A$ we can observe that $A$ differs in exactly $3$ positions from $B$ and from $\bE^{-2} B$.
For the sequence $B$ we can observe that $B$ differs in exactly $3$ positions from $C$ and from $\bE^4 C$.
For the sequence $C$ we can observe that $C$ differs in exactly $3$ positions from $\bE^3 A$ and from $\bE^4 A$.

The subtypes with their cyclic orders will be ordered in a way that any two consecutive ones (including the last and the first) differ in
exactly one sequence as follows by their cyclic order with their first shift vector and the last shift vector.
$$
\begin{array}{lccc}
& \text{cyclic~order} & \text{first~shift~vector} & \text{last~shift~vector}\\
& [AAB] & [0,0,0] & [0,0,2] \\
& [BAB] & [0,0,2] & [0,0,0] \\
& [BAC] & [0,0,0] & [0,1,0] \\
& [CAC] & [0,1,0] & [0,0,0] \\
& [CBC] & [0,0,0] & [0,0,4] \\
& [BBC] & [0,0,4] & [0,0,0] \\
& [ABC] & [0,0,0] & [0,0,4] \\
& [AAC] & [0,0,4] & [0,0,0] \\
\end{array}
$$

Now we need to insert the $81$ SDSs of $[AAA]$, $[BBB]$, and $[CCC]$ into the list with 648 SDSs. Each one has $27$ SDSs and each SDS can be
interleaved using one of three possible shift vectors. The $27$ SDSs of $[AAA]$ are inserted between SDSs of $[AAB]$ with
the following rule. Assume the SDS obtained from the shift vector $[0,i,j]$ has to be inserted. It will be inserted
between the two consecutive SDSs of $[AAB]$ with the shift vectors $[0,i,j]$ and $[0,i,j+2]$.
The $27$ SDSs of $[BBB]$ are inserted between SDSs of $[BBC]$ with
the following rule. Assume the SDS obtained from the shift vector $[0,i,j]$ has to be inserted. It will be inserted
between the  two consecutive SDSs of $[BBC]$ with the shift vectors $[0,i,j]$ and $[0,i,j+4]$.
The $27$ SDSs of $[CCC]$ are inserted between SDSs of $[CAC]$ with
the following rule. Assume the SDS obtained from the shift vector $[0,i,j]$ has to be inserted. It will be inserted
between the  two consecutive SDSs of $[CAC]$ with the shift vectors $[0,i+3,j]$ and $[0,i+4,j]$.

Altogether we have $729$ SDSs of period $27$ and now we should complete their order to satisfy the conditions of
Theorem~\ref{thm:self_dualST_NB} and also Properties (P1) and (P2).
Property (P1) is satisfied since the ordering we done so far is cyclic.
Property (P2) is satisfied since each one of the three sets of coordinates for the interleaving
is chosen not to be a pivot at least once in the $8$ cyclic orders of the subtypes that are not $[AAA]$, $[BBB]$,
and $[CCC]$. To satisfy the conditions of Theorem~\ref{thm:self_dualST_NB} we would like first that the
sequence that starts with $9$ \emph{zeroes} will be followed by the sequence that starts with $8$ \emph{zeroes}
followed by a $2$. For this, we first pull out $[AAA]$ with shift vectors $[0,0,0]$ and $[0,0,4]$ outside
and observe the current ordering of the following consecutive shift vectors in the subtype $[AAC]$:
$[0,0,4]$, $[0,0,0]$, $[0,0,5]$, $[0,0,1]$,
and $[0,0,6]$. We make the following change in these $5$ shift vectors with the two shift vectors of $[AAA]$.
These are represented by the following $7$ pairs which will be consecutive in the list:
$$
\{ [AAC],[0,0,4]\}, \{[AAC],[0,0,0]\}, \{[AAA],[0,0,4]\}, \{[AAC],[0,0,1]\},
$$
$$
\{ [AAC],[0,0,5]\}, \{[AAA],[0,0,0]\}, \{[AAC],[0,0,6]\}.
$$
Now, $\{[AAA],[0,0,0]\}$ and $\{[AAC],[0,0,6]\}$ start with $9$ consecutive \emph{zeroes} and $8$ consecutive
\emph{zeroes} followed by a $2$, respectively, to complete the requirements of Theorem~\ref{thm:self_dualST_NB}
and the required seed for $p=3$.

\section{Seed initial arrangement for $p=5$}
\label{sec:seed_p5}

The initial arrangement for $p=5$ and $n=2$ was found using the arrangement of differences,
There are $125$ SDSs in this case, and we consider first the sequences of differences, of length $5$ that start with $00$ (whose sum is $1$ modulo $5$)
in lexicographic order. There are 20 such sequences. We continue and add a constant $1$ to the entries of these sequences to obtain
the sequences of differences that start with $11$. Cyclic sequences that already appear in the list of sequences starting with $00$
are ignored. To these sequences we add again the constant $1$ to obtain sequences of differences that start with $22$ and ignore those
that already appeared in the lists that start with $00$ or $11$. The same is done with sequences starting with $33$ and $44$.
Five lists are obtained and they contain $85$ sequences out of the $125$ sequences.
The list of $14$ sequences starting with $44$ are ordered first. The same order, with two more sequences, is used for the sequences
starting with $33$, and for the $17$ sequences starting with $22$, the $18$ sequences starting with $11$, and finally $20$ sequences
starting with $00$. There are an additional $40$ sequences of differences.
The ones starting with either $02$, or $13$, or $24$, or $30$, or $41$, respectively ($2$ sequences in each list).
These are merged within the lists of $00$, $11$, $22$, $33$, and $44$, respectively.
We also have the ones starting with either $01$, or $12$, or $23$, or $34$,
or $40$ ($4$ sequences in each list).
The ones starting with either $03$, or $14$, or $20$, or $31$, or $42$ ($2$ sequences in each list). These sequences
are ordered in five lists of length $6$. Each such list is used to merge two of the previous five lists.
Having now a cyclic ordering, a list of SDSs satisfying the properties of the initial arrangement are obtained from
this list of sequences of differences. Property (P1) and the requirement for the first and the last SDSs in the list
are satisfied easily by considering the SDS that starts with $5$ \emph{zeroes} as the first SDS in the list.
Property (P2) is satisfied by itself, i.e., there are $125$ SDSs in this initial arrangement, so it was reasonable that
for each one of the first $5$ coordinates, there will be a pair of consecutive SDSs that differ in this coordinate.
The final list of $125$ SDSs satisfying the initial arrangement and properties (P1) and (P2) is presented here (the first
$15$ digits of each SDS are presented).

\begin{example}
$$
\begin{array}{c}
   000001111122222 \\
   020001211123222 \\
   020001311124222 \\
   030001411120222 \\
   040001011121222 \\
   140002011131222 \\
   130002411130222 \\
   110002211133222 \\
   120002311134222 \\
   220003311144222 \\
   240003011141222 \\
   210003211143222 \\
   230003411140222 \\
   233003441140022 \\
   033001441120022 \\
   433000441110022 \\
   434000401110122 \\
   432000431110422 \\
   431000421110322 \\
   431100422110332 \\
   430100412110232 \\
   434100402110132 \\
   433100442110032 \\
   423100342114032 \\
   423110342214033 \\
\end{array}
\hspace{-0.185cm}
\begin{array}{c}
   423210343214043 \\
   403210143212043 \\
   433210443210043 \\
   433410440210013 \\
   433010441210023 \\
   433110442210033 \\
   434110402210133 \\
   434310404210103 \\
   434410400210113 \\
   430410410210213 \\
   430110412210233 \\
   430010411210223 \\
   430040411010221 \\
   433040441010021 \\
   432040431010421 \\
   432020431310424 \\
   032021431320424 \\
   032011431220423 \\
   032001431120422 \\
   032041431020421 \\
   002041131022421 \\
   002141132022431 \\
   003141142022031 \\
   303144142002031 \\
   103142142032031 \\
\end{array}
\hspace{-0.185cm}
\begin{array}{c}
   104142102032131 \\
   100142112032231 \\
   101142122032331 \\
   102142132032431 \\
   112142232033431 \\
   114142202033131 \\
   110142212033231 \\
   120142312034231 \\
   122142332034431 \\
   121142322034331 \\
   121442320034311 \\
   101442120032311 \\
   141442020031311 \\
   141042021031321 \\
   141242023031341 \\
   141202023131342 \\
   141102022131332 \\
   141002021131322 \\
   141402020131312 \\
   341404020101312 \\
   342404030101412 \\
   302404130102412 \\
   302434130402410 \\
   302414130202413 \\
   312414230203413 \\
\end{array}
\hspace{-0.185cm}
\begin{array}{c}
   322414330204413 \\
   332414430200413 \\
   342414030201413 \\
   442410030211413 \\
   432410430210413 \\
   412410230213413 \\
   422410330214413 \\
   022411330224413 \\
   042411030221413 \\
   032411430220413 \\
   030411410220213 \\
   330414410200213 \\
   230413410240213 \\
   231413420240313 \\
   233413440240013 \\
   233013441240023 \\
   232013431240423 \\
   231013421240323 \\
   230013411240223 \\
   220013311244223 \\
   220023311344224 \\
   220123312344234 \\
   210123212343234 \\
   200123112342234 \\
   230123412340234 \\
\end{array}
\hspace{-0.185cm}
\begin{array}{c}
   230223413340244 \\
   230323414340204 \\
   230423410340214 \\
   230023411340224 \\
   231023421340324 \\
   231423420340314 \\
   231223423340344 \\
   231323424340304 \\
   232323434340404 \\
   232023431340424 \\
   232423430340414 \\
   232403430140412 \\
   230403410140212 \\
   234403400140112 \\
   234413400240113 \\
   234433400440110 \\
   334434400400110 \\
   334424400300114 \\
   334414400200113 \\
   334404400100112 \\
   304404100102112 \\
   304004101102122 \\
   300004111102222 \\
   200003111142222 \\
   100002111132222 \\
\end{array}
$$
\end{example}

\section*{Acknowledgement}


%


\end{document}